\documentclass[hidelinks, 12pt]{article}
\usepackage{geometry}
\usepackage{setspace}
\usepackage[utf8]{inputenc}
\usepackage{xcolor}
\usepackage{bm}
\usepackage{amsfonts}
\usepackage{amsmath}
\usepackage{amssymb}
\usepackage{amsthm}
\usepackage{graphicx}
\usepackage{xfrac}
\usepackage{multirow}
\usepackage{pgfplots}
\usepackage{tikz}
\usepackage{dcolumn}
\usepackage{threeparttable}
\usepackage{lscape}
\usepackage{tabularx}
\usepackage{booktabs}
\usepackage{caption}
\usepackage{subcaption}
\usepackage{siunitx}
\usepackage{adjustbox}
\usepackage{accents}
\usepackage{comment}
\usepackage{epstopdf}
\usepackage{hyperref}
\usepackage[noabbrev]{cleveref}
\usepackage{tasks}
\usepackage{appendix}
\usepackage{mathtools}
\usepackage{fancyhdr}
\usepackage[round]{natbib}

\usepackage [english]{babel}
\usepackage [autostyle, english = american]{csquotes}
\MakeOuterQuote{"}
\usepackage{authblk}
\usepackage{blindtext}

\usepackage{datetime}

\newdateformat{monthyeardate}{%
  \monthname[\THEMONTH], \THEYEAR}

\usepackage{float}
\usepackage{rotating}
\usepackage{enumitem,amssymb}
\usepackage{soul}
\usepackage{bbm}

\usepackage{times}

\pgfplotsset{compat=1.15}

\usepackage{footmisc}

\newtheorem{proposition}{Proposition}

\theoremstyle{definition}

\theoremstyle{definition}
\newtheorem{example}{Example}

\theoremstyle{definition}

\crefname{prop}{Proposition}{Propositions}
\crefname{mydef}{Definition}{Definitions}
\crefname{lemma}{Lemma}{Lemmas}
\crefname{definition}{Definition}{Definitions}
\crefname{theorem}{Theorem}{Theorems}
\crefname{proposition}{Proposition}{Propositions}
\crefname{corollary}{Corollary}{Corollaries}
\crefname{assumption}{Assumption}{Assumptions}
\crefname{claim}{Claim}{Claims}
\crefname{section}{Section}{Sections}
\crefname{figure}{Figure}{Figures}
\crefname{exmp}{Example}{Examples}
\crefname{observation}{Observation}{Observations}

\newcommand{\reals}{\mathbb{R}}

\DeclareMathOperator*{\argmax}{arg\,max}

\makeatletter
\DeclareRobustCommand\citepos													
  {\begingroup\def\NAT@nmfmt##1{{\NAT@up##1's}}%
   \NAT@swafalse\let\NAT@ctype\z@\NAT@partrue
   \@ifstar{\NAT@fulltrue\NAT@citetp}{\NAT@fullfalse\NAT@citetp}}

\pretocmd{\NAT@citex}{%
  \let\NAT@hyper@\NAT@hyper@citex
  \def\NAT@postnote{#2}%
  \setcounter{NAT@total@cites}{0}%
  \setcounter{NAT@count@cites}{0}%
  \forcsvlist{\stepcounter{NAT@total@cites}\@gobble}{#3}}{}{}
\newcounter{NAT@total@cites}
\newcounter{NAT@count@cites}
\def\NAT@postnote{}

\def\NAT@hyper@citex#1{
  \stepcounter{NAT@count@cites}%
  \hyper@natlinkstart{\@citeb\@extra@b@citeb}#1%
  \ifnumequal{\value{NAT@count@cites}}{\value{NAT@total@cites}}
    {\if*\NAT@postnote*\else\NAT@cmt\NAT@postnote\global\def\NAT@postnote{}\fi}{}%
  \ifNAT@swa\else\if\relax\NAT@date\relax
  \else\NAT@@close\global\let\NAT@nm\@empty\fi\fi								
  \hyper@natlinkend}
\renewcommand\hyper@natlinkbreak[2]{#1}

\patchcmd{\NAT@citex}															
  {\ifNAT@swa\else\if*#2*\else\NAT@cmt#2\fi
   \if\relax\NAT@date\relax\else\NAT@@close\fi\fi}{}{}{}
\patchcmd{\NAT@citex}
  {\if\relax\NAT@date\relax\NAT@def@citea\else\NAT@def@citea@close\fi}
  {\if\relax\NAT@date\relax\NAT@def@citea\else\NAT@def@citea@space\fi}{}{}
\patchcmd{\NAT@cite}{\if*#3*}{\if*\NAT@postnote*}{}{}
\makeatother

\title{The (No) Value of Commitment}
\author{Nathan Hancart}
\date{\monthyeardate\today}

\begin{document}

\maketitle

\begin{abstract}
    \noindent I provide a sufficient condition under which a principal does not benefit from committing to a mechanism in economic models represented by a maximisation problem under constraints. These problems include mechanism design, principal-agent models or sender-receiver games. In principal-agent problems, this condition holds if the agent has a finite strategy space and the principal's value function is continuous in the mechanism. 
\end{abstract}

\section{Introduction}

Commitment plays an important role in many economic models. The general insight of economic theory is that the value of commitment is positive: if a principal has commitment, he can replicate any action he would play without commitment. Moreover, commitment plays a key role in many standard tools used in economic theory such as the revelation principle \citep{myerson1982, bester_strausz2000, doval_skreta2022}. However, commitment is usually a strong assumption and is sometimes hard to justify. Even if we are willing to assume commitment, knowing that the principal best-replies to the strategy of the agent restricts the set of strategies the modeller has to look at.

To fix ideas, I introduce the setup I will be looking at directly, then will discuss the results.

A principal chooses a mechanism $\alpha\in A$. The principal can also choose a strategy $\sigma\in \Sigma$, subject to the constraint that $\sigma\in C(\alpha)\subseteq \Sigma$. The strategy $\sigma$ could be the strategy of agents in mechanism design problems or a receiver in a sender-receiver game. Denote by $v(\alpha,\sigma)$ the principal's payoffs.

When the principal commits to $\alpha$, it gets\begin{equation}
    V(\alpha)=\max_\sigma v(\alpha,\sigma)\quad \text{s.t.}\;\sigma\in C(\alpha).\tag{$\mathcal{V}(\alpha)$}
\end{equation}

The principal does not benefit from commitment if there is $\alpha^*\in \argmax_\alpha V(\alpha)$ and $\sigma^*\in C(\alpha^*)$ such that $v(\alpha^*,\sigma)\geq v(\alpha,\sigma)$ for all $\alpha\in A$.

In this note, I show that there is no value of commitment when a condition I call stable preferred strategies hold at the optimal mechanism. It holds if for every nearby mechanism, one of the strategies that attains the principal's optimal payoffs is still feasible. In the context of principal-agent problems, where the constraint set $C(\alpha)$ is the best-reply correspondence of the agent, I show that the continuity of the value function $V(\alpha)$ implies the stable preferred strategies property. This result generalises the result from \citet{benporath_et_al2021} who show that there is no value of commitment whenever the principal is indifferent amongst the best-replies of the agent at the optimal mechanism: when this condition is satisfied, the value function is continuous. I also link this result to \citepos{lipnowski2020} result that also shows that continuity of the value function implies the equivalence between cheap-talk and Bayesian persuasion in sender-receiver games.

In an example, I also show that the simple type dependence assumption on the agent's preferences is key to the result of no value of commitment in games with evidence disclosure introduced in \citet{benporath_et_al2019}. In particular, it is not true that in evidence games where the principal's payoff is $v(a,\theta)=\nu(\theta)u(a,\theta)+\psi(a)$ where $a$ is the action taken by the principal, $\theta$ the type of the agent and $u(a,\theta)$ his payoffs, the principal does not benefit from commitment.

\section{Results}

I assume throughout that $A$ is a compact convex subset of $\reals^n$. The principal can also choose a strategy $\sigma\in \Sigma$, also a closed convex subset of $\reals^n$, subject to the constraint that $\sigma\in C(\alpha)\subseteq \Sigma$, where $C(\cdot)$ is a non-empty, closed, convex and upper hemicontinuous correspondence.

The constraint $C(\alpha)$ can be a best-reply correspondence, a set of correlated equilibria or some other constraint set. However, it cannot be a set of Nash Equilibria as these are not necessarily convex. I also note that $\sigma$ can incorporate some elements of the mechanism the principal has the power to commit to. The mechanism $\alpha$ can be a strategy of the principal in a mechanism design problem or an information structure in a sender-receiver game.

I assume that $v(\alpha,\sigma)$ is concave in $\alpha$, linear in $\sigma$ and jointly continuous. 

I denote by $V^*=\max_\alpha V(\alpha)$ the maximal value attainable by the principal.

The principal has \textit{stable preferred strategies} at $\alpha^*$ if there is  $\epsilon>0$ such that for all $\alpha\in B_\epsilon(\alpha^*)$,\begin{equation}\label{eq_cond2}
\text{there is }\sigma\in C(\alpha)\text{ such that }v(\alpha^*,\sigma)=V^*. 
\end{equation}

In words, the principal has stable preferred strategies at $\alpha^*$ if for every nearby mechanism, one of the strategies that attains the principal's optimal payoffs is still feasible. 

\begin{proposition}\label{main_result}
    Let $\alpha^*\in \max V(\alpha)$. If the principal has stable preferred strategies at $\alpha^*$, there is no value of commitment.
\end{proposition}

The proof builds on and generalises the proof of \citet{benporath_et_al2021}, Theorem 4.

\begin{proof}
    Let $\Tilde{A}$ a closed ball of $A$ containing $\alpha^*$ where condition (\ref{eq_cond2}) holds. The set $\Tilde{A}$ is compact and convex. Let $V^*=\max_\alpha V(\alpha)$.

    Let $\Sigma^*=\{\sigma: \sigma\in C(\alpha)\text{ for some $\alpha\in \Tilde{A}$ and }v(\alpha^*,\sigma)=V^*\}$. By condition (\ref{eq_cond2}), $\Sigma^*$ is non-empty. Because $C$ is upper hemicontinuous and $v$ is continuous in $\sigma$, $\Sigma^*$ is also closed and therefore compact. Because $C$ is convex-valued and $v(\alpha,\cdot)$ is linear, $\Sigma^*$ is convex. By definition, $\Sigma^*\cap C(\alpha)$ is non-empty for all $\alpha\in \Tilde{A}$ and is also closed and convex. Similarly, $\argmax_{\alpha\in \Tilde{A}}v(\alpha,\sigma)$ is non-empty, closed and convex valued. Moreover, both $\Sigma^*\cap C(\alpha)$ and $\argmax_{\alpha\in \Tilde{A}}v(\alpha,\sigma)$ are upper hemicontinuous in $\alpha$ and $\sigma$ respectively.

     By the Kakutani fixed-point theorem, there exists a fixed point of the correspondence $C(\alpha)\cap \Sigma^*\times \argmax_{\alpha\in \Tilde{A}}v(\alpha,\sigma)$, $(\hat{\alpha},\hat{\sigma})$. We want to show that $\hat{\sigma}\in C(\hat{\alpha})$, $\hat{\alpha}\in \argmax_{\alpha\in A} v(\alpha,\hat{\sigma})$ and $V^*=v(\hat{\alpha},\hat{\sigma})$.

     By definition, $\hat{\sigma}\in C(\hat{\alpha})$.

     Because $\hat{\sigma}\in \Sigma^*\cap C(\hat{\alpha})$, we have $V^*=v(\alpha^*,\hat{\sigma})\geq v(\hat{\alpha},\hat{\sigma})$. Because $\alpha^*\in \Tilde{A}$, we also have $v(\hat{\alpha},\hat{\sigma})\geq v(\alpha^*,\hat{\sigma})$. Therefore, $v(\hat{\alpha},\hat{\sigma})= v(\alpha^*,\hat{\sigma})$.

     For any $\alpha\in A$, $(1-\epsilon)\alpha^*+\epsilon \alpha\in \Tilde{A}$ for $\epsilon>0$ small enough. Because $\hat{\alpha}\in \argmax_{\alpha\in \Tilde{A}}v(\alpha,\hat{\sigma})$, we get $$
    v(\alpha^*,\hat{\sigma})= v(\hat{\alpha},\hat{\sigma})\geq v((1-\epsilon)\alpha^*+\epsilon \alpha,\hat{\sigma}).
    $$Using the concavity of $v$ and rearranging, we obtain $ v(\alpha^*,\hat{\sigma})=v(\hat{\alpha},\hat{\sigma})\geq v(\alpha,\hat{\sigma})$. Therefore, $\hat{\alpha}\in \argmax_{\alpha\in A} v(\alpha,\hat{\sigma})$.
\end{proof}

\cref{main_result} gives us a tool to verify whether the principal has no value of commitment. We can obtain a solution that is easier to verify if we put more assumptions on $C$ as we do in the next section.

\section{Principal-agent problems}

Suppose the situation we want to model is a principal-agent problem. Let $u(\alpha,\sigma)$ the utility function of the agent and $BR(\alpha)$ be the set of best-replies to $\alpha$. Here, we have $BR(\alpha)=C(\alpha)$. I also assume that $\Sigma=\Delta(S)^n$ for some finite set $S$. I assume that $u(\alpha,\sigma)$ and $v(\alpha,\sigma)$ are linear in $\sigma$ and jointly continuous. If these assumptions are satisfied, we are in a \textit{principal-agent problem}.

\begin{proposition}\label{prop_cont}  
    Suppose we are in a principal-agent problem. Let $\alpha^*\in\argmax_\alpha V(\alpha)$. If $V(\alpha)$ is continuous at $\alpha^*$, then the principal has stable preferred strategies at $\alpha^*$.
\end{proposition}

\begin{proof}
   Let $\Sigma^*=\{\sigma:v(\alpha^*,\sigma)=V^*\}$.
    
    Suppose it is not the case, i.e., for all $\epsilon>0$, there is $\alpha_\epsilon\in B_\epsilon(\alpha^*)$ such that $$
    \text{for all}\,\sigma\in \Sigma^*,\; u(\alpha_\epsilon,\sigma^*)< u(\alpha_\epsilon,\sigma), \text{for some}\,\sigma\notin \Sigma^*.
    $$
    Take a sequence of $\alpha_\epsilon$ and a principal-preferred best-reply $\sigma_\epsilon\in BR(\alpha_\epsilon)$. By hypothesis, $\sigma_\epsilon\notin \Sigma^*$. Without loss, we can take $\sigma_\epsilon$ in pure strategy and because it does not belong to $\Sigma^*$, the pure strategy $\sigma_\epsilon$ does not belong to the support of $\Sigma^*$. This implies that $\sigma_\epsilon\rightarrow \Tilde{\sigma}\notin \Sigma^*$ as $\epsilon\rightarrow 0$ (possibly taking a convergent subsequence and using that there is a finite number of pure strategies). By continuity of $u$, we have $\Tilde{\sigma}\in BR(\alpha^*)$. But then by continuity of $V$, we have $V(\alpha_\epsilon)=v(\alpha_\epsilon,\sigma_\epsilon)\rightarrow V^*=v(\alpha^*,\Tilde{\sigma})$ for $\sigma\notin \Sigma^*$. A contradiction.\end{proof}

Continuity of the value does not generally imply that the stable preferred strategy condition holds. A key assumption we used in the proof here is that we can focus without loss of generality on pure strategies and that converging sequences of pure strategies have to be constant at some point in the sequence. That would not be the case in an infinite strategy space.

\citet{lipnowski2020} (conjectured by \citet{forges2020}) shows that in sender-receiver games with a finite number of actions and states, an optimal information structure can be implemented by cheap-talk when the value function that takes as input posterior beliefs is continuous. The model presented here allows for modelling sender-receiver games. Let $\Theta$ be a set of state, $M$ a message space and $S$ a finite action space. The strategy of the principal is an information structure $\alpha:\Theta\rightarrow \Delta M$ and the strategy of the agent is $\sigma:M\rightarrow\Delta S$. The posterior distribution given message $m$ varies continuously in the information structure $\alpha$. Therefore \cref{prop_cont} implies \citepos{lipnowski2020} result when restricting attention to sender-receiver games.

\citet{benporath_et_al2021} show that whenever at the optimal mechanism $\alpha$, the principal is indifferent across all the best-replies of the agent, then the principal does not benefit from commitment. If a game has this property, they call this game \textit{aligned}. The following result shows that if the game is aligned, then the value function is continuous.

\begin{proposition}
    Suppose we are in a principal-agent problem. If for all $\sigma\in BR(\alpha)$, $v(\alpha,\sigma)=V(\alpha)$, then $V(\alpha)$ is continuous at $\alpha$.
\end{proposition}

\begin{proof}
    By the upper hemicontinuity of the best-reply correspondence and the fact that $\Sigma=(\Delta S)^n$ and $u(\alpha,\cdot)$ is linear, there is a open ball around $\alpha$, $B_\epsilon(\alpha)$ such that for all $\alpha'\in B_\epsilon(\alpha)$, $BR(\alpha')\subseteq BR(\alpha)$. Therefore, for any sequence $(\alpha_n,\sigma_n)$ with $\alpha_n\in B_\epsilon(\alpha)$ and $\sigma_n\in BR(\alpha_n)$, we have (possibly taking a convergent subsequence), $\sigma_n\rightarrow\sigma \in BR(\alpha)$. Therefore, $v(\alpha_n,\sigma_n)\rightarrow V(\alpha)$.
\end{proof}

This observation was also made by \citet{lipnowski2020} in the context of sender-receiver games. Most results on the value of commitment in the mechanism design with evidence literature have games that are aligned. Some examples are \citet{glazer_rubinstein2004, glazer_rubinstein2006}, \citet{benporath_et_al2019} or \citet{hart_et_al2017} (see \citet{benporath_et_al2021} for a discussion). \citet{kamenica_lin2024} establish a similar result in sender-receiver games.

The conditions of aligned games and continuous value function differ when there is a best-reply of the agent, say $\sigma'$ that makes the principal strictly worse off at the optimal mechanism $\alpha^*$. To guarantee that the value function is continuous at $\alpha^*$, it must also be the case that $\sigma'$ is not strictly preferred if we perturb $\alpha^*$ a little bit. 

Below, I provide an example of a mechanism design problem with evidence and moral hazard where the value function is continuous and therefore the principal does not benefit from commitment. Noteworthy features of this examples are that the agent's payoff is non-linear in the mechanism and the principal is not indifferent across best-replies of the agent at the optimum. I then show an example of a game where there is a discontinuity in the value function and there is value to commitment. At the same time, I show that the result on the value of commitment in \citet{benporath_et_al2019} does not extend beyond simple type dependence.\footnote{In a principal-agent problem with private information $\theta\in \Theta$, a payoff function satisfies simple type dependence if for all $\theta\in \Theta$, $u(a,\theta)=u(a)$ or $=-u(a)$ for some function $u(a)$.}

\begin{example}
    This example is a variation over mechanism with Dye evidence \citep{dye1985}. The principal must make an investment decision, $a\in [-2,2]$. The agent can make an investment in their skills $x=0,1$. In addition, the agent has a private information about their productivity $\theta\in\Theta\subset [-1,1]$. With some probability, the agent can certify his productivity. With complement probability, he can only disclose that he did not receive evidence, denoted by $\emptyset$. Each type is pair of productivity and evidence to disclose, either $\{\{\theta\},\emptyset\}$ or $\{\emptyset\}$.\footnote{Given the payoffs we will assume, cheap-talk cannot help the principal in this mechanism design problem.} Therefore a type is an element of $\Theta\times \{e_1,e_0\}$, where $e_1$ indicates that the agent can produce evidence and $e_0$ that he cannot. The set of evidence to disclose is $M=\Theta\cup\{\emptyset\}$. The prior distribution on types is denoted by $\mu\in \Delta(\Theta\times \{e_1,e_0\})$.

    The payoff of the principal is $v(a,x,\theta)=\theta\cdot (1+x)\cdot a -\frac{a^2}{2}$. The payoff of the agent is $u(a,x)=xa+(1-x)\big(\mathbbm{1}[a<0]\cdot\frac{a}{2}+\mathbbm{1}[a\geq 0]\cdot a\big)$. Therefore when $a\geq 0$, the agent is indifferent between action $x=1$ and $x=0$. When $a<0$, the agent strictly prefers action $x=0$. The agent always strictly prefers a higher $a$.
    
    The interpretation is that the principal can invest in the production of two goods where positive $a$ indicates the first good and negative $a$ the second good. The agent wants the principal to invest as much as possible in the first good as little as possible in the second. Positive productivity $\theta$ is associated with positive productivity in the first good and negative $\theta$ is associated in positive productivity in the second good. On top of that, the agent can choose whether to improve his productivity in the dimension he is good at by choosing $x=1$ or $x=0$. He strictly prefers to not invest if he is assigned to the task he does not like ($a<0$) but is otherwise indifferent.

    A strategy for the agent is a mapping $\sigma:\Theta\times \{e_1,e_0\}\rightarrow \Delta(M\times \{0,1\})$ with the restriction that $\sigma(\theta,x\vert\theta,e)>0\Rightarrow e=e_1$ and $\sigma(\theta',x\vert\theta,e)=0$ for all $\theta'\neq \theta$. That is, only types that have evidence can produce evidence and evidence has to be correct.

    A strategy for the principal is a mapping $\alpha:M\rightarrow[-2,2]$. Note that I assume that the investment decision $x$ is not observed by the principal.

    For a given evidence $e$, the principal's preferred best-reply is if $\theta>0$, to play $x=1$ whenever $\alpha(m)\geq 0$ and $x=0$ otherwise. If $\theta<0$, to play $x=0$ for all $\alpha(m)$.

    Given $\alpha$, the agent always chooses the message that maximises his probability of being accepted.
    
    For type $(\theta,e_1)$ with $\theta>0$, the value to the principal is $$
    \theta \max\{\alpha(\theta),\alpha(\emptyset)\}\big(1+\mathbbm{1}[\max\{\alpha(\theta),\alpha(\emptyset)\}\geq 0]\big)-\frac{\max\{\alpha(\theta),\alpha(\emptyset)\}^2}{2}.
    $$
    For type $(\theta,e_1)$ with $\theta<0$, the value to the principal is $$
    \theta \max\{\alpha(\theta),\alpha(\emptyset)\}-\frac{\max\{\alpha(\theta),\alpha(\emptyset)\}^2}{2}.
    $$
    For type $(\theta,e_0)$ with $\theta>0$, the value to the principal is $$
    \theta \alpha(\emptyset)\big(1+\mathbbm{1}[\alpha(\emptyset)\geq 0]\big)-\frac{\alpha(\emptyset)^2}{2}.
    $$
    For type $(\theta,e_0)$ with $\theta<0$, the value to the principal is $$
    \theta \alpha(\emptyset)-\frac{\alpha(\emptyset)^2}{2}.
    $$
    All these functions are continuous in $\alpha$ and therefore for any distribution $\mu$, the value function $V(\alpha)$ is continuous. 
\end{example}

The following example considers a generalisation of the payoff function introduced in \citet{benporath_et_al2019} that does not assume simple type dependence. 

\begin{example}
    Let $\Theta=\{\theta_1,\theta_2\}$, $M=\{m_1,m_2\}$, $M(\theta_1)=\{m_1\}$, $M(\theta_2)=\{m_1,m_2\}$ and $A=\{1,2,3\}$. Take some $\epsilon\in (\frac{2}{9},\frac{1}{3})$. The prior probability on type $\theta_1$ is $\mu$.

The utility of the agent is $u(a,\theta_1)=a$ and $$u(a,\theta_2)=\begin{cases}
    0&\text{if }a=1\\
    1-\epsilon&\text{if }a=2\\
    1&\text{if }a=3
\end{cases}.$$The utility of the principal is $v(a,\theta)=\nu(\theta)u(a,\theta)+\psi(a)$ with $\nu(\theta_1)=1$, $\nu(\theta_2)=-3$ and $\psi(1)=1$ and $\psi(2)=\psi(3)=2$. The principal's payoff is increasing in $a$ when $\theta=\theta_1$ and decreasing when $\theta=\theta_2$.

A strategy for the agent is a mapping $\sigma:\Theta\rightarrow\Delta M$ with $\sigma(m\vert\theta)>0\Rightarrow m\in M(\theta)$. Only type $\theta_2$ has a real choice of strategy and we will denote by $\sigma$ the probability of choosing message $m_2$. A strategy for the principal is a mapping $\alpha:M\rightarrow\Delta A$.

\textbf{Claim: }In the game described above, there is value to committing to $\alpha$.

    Suppose there is no value of commitment. At the optimum, if $\sigma(m_2\vert\theta_2)>0$, then the best-reply to $m_2$ is to play $a=1$ as only $\theta_2$ can send message $m_2$. Given this best-reply, type $\theta_2$ is only willing to send $m_2$ with positive probability if $\alpha(1\vert m_1)=1$, otherwise type $\theta_2$ strictly prefers to send $m_1$. Therefore the principal plays the same action independently of the message received.

    We conclude that without commitment, the best the principal can do is to use the same action independently of the message received. 

    Note that given the assumption on $\epsilon$, there is always a prior $\mu$ such that playing $a=2$ is a best-reply ex-ante for the principal:\begin{align*}
        &4\mu+(1-\mu)(-1+3\epsilon)\geq 2\mu +(1-\mu),\\
        &4\mu+(1-\mu)(-1+3\epsilon)\geq 5\mu +(1-\mu)(-1).
    \end{align*}

    Now I show that there is a strategy $\alpha$ that can do better than that. Let $\alpha(1\vert m_1)=\epsilon$, $\alpha(3\vert m_1)=1-\epsilon$ and $\alpha(2\vert m_2)=1$. The strategy of the agent is $\sigma(m_2\vert\theta_2)=1$. 

    This strategy gives strictly higher profits to the principal and is incentive-compatible for the agent. 

    The IC constraint of $\theta_2$ must satisfy$$
1-\epsilon\geq \epsilon\cdot0+(1-\epsilon)\cdot1.
    $$

    I show that the profits must be higher than playing $a=2$ with probability one:\begin{gather*}
        \mu\epsilon\cdot 2+\mu(1-\epsilon)\cdot 5+(1-\mu)(-1+3\epsilon)> \mu 4+(1-\mu)(-1+3\epsilon)\\
        \Leftrightarrow \epsilon< 1/3.
    \end{gather*}
        Since playing $a=2$ is strictly optimal ex-ante for some $\mu$, there is a mechanism where the principal does not best-reply and gives strictly higher profits against all the strategies where the principal best-replies.

To see there is a discontinuity, observe that if the principal decreases $\alpha(2\vert m_1)$ and increases $\alpha(1\vert m_2)$, then type $\theta_2$ has a strict incentive to choose $m_1$ leading to a discrete drop in profits.

\end{example}

\newpage

\bibliographystyle{agsm}
\bibliography{bib.bib} 

\end{document}